\documentclass[12pt]{iopart}
\usepackage{iopams}
\usepackage{setstack}
\usepackage{graphicx}
\usepackage{amsthm}

\newtheorem{proposition}{Proposition}

\theoremstyle{remark}

\begin{document}
\newcommand{\real}{\textrm{Re}\:}
\newcommand{\ddp}{\partial}
\newcommand{\sto}{\stackrel{s}{\to}}
\newcommand{\supp}{\textrm{supp}\:}
\newcommand{\wto}{\stackrel{w}{\to}}
\newcommand{\ssto}{\stackrel{s}{\to}}
\newcounter{foo}
\providecommand{\norm}[1]{\lVert#1\rVert}
\providecommand{\abs}[1]{\lvert#1\rvert}

\title{The few--body universality is not exact for more than three particles}

\author{Dmitry K. Gridnev}

\address{FIAS, Ruth-Moufang-Stra{\ss}e 1, D--60438 Frankfurt am Main, Germany}
\ead{gridnev@fias.uni-frankfurt.de}
\begin{abstract}
In the literature it is conjectured that the ground state energies of three and four bosons are universally related for all 
pair--interactions given that two bosons have a zero energy resonance and no negative energy bound states. Here it is proved analytically that such relation cannot be exact. 
\end{abstract}



\section{Introduction}\label{sec:1}

In 1970 V. Efimov predicted \cite{vefimov} a remarkable phenomenon now called the Efimov effect, which can be stated as follows.
If the negative continuous spectrum of a three-particle Hamiltonian $H$ is empty but
at least two of the particle pairs have a resonance at zero energy then $H$ has an infinite number of bound states with the energies $E_n <0$ for $n = 1, 2, \ldots$. 
If the particles are bosons then the limit ratio of the energy levels satisfies the equation $\lim_{n \to \infty} E_n / E_{n+1} = e^{2\pi/s_0}$, where the 
constant $s_0$ is universal and obeys the equation \cite{sobolev}
\begin{equation}
 s_0 = \frac{8}{\sqrt{3}} \frac{\sinh(\pi s_0/6)}{\cosh(\pi s_0/2)} . 
\end{equation}
Locating the root of this equation yields $s_0 \simeq 1.006$. In particular, this result holds for \textit{all} pair--interactions, which are bounded and fall off faster than 
$(1 + r^2)^{-1}$, therefore, one speaks of \textit{three--body universality} \cite{braatenphysrep}.  The first sketch of mathematical proof
of the Efimov effect was done by
L. D. Faddeev shortly after V. Efimov told him about his discovery \cite{efimovprivate}. The first published proof, which was not completely rigorous, 
appeared in \cite{amadonoble}. Later D. R. Yafaev \cite{yafaev} basing on the Faddeev's idea presented a complete proof thus turning 
the Efimov effect into a mathematical fact. 
In \cite{ovchinnikov,fonseca} one finds other proofs by different methods. Tamura \cite{tamura} generalized the proof to the case when pair--potentials can take both signs. 
In \cite{wang} the author claimed having generalized the result in \cite{yafaev,sobolev} to the case of three clusters but the proof in \cite{wang}
contains a mistake \cite{wangwrong}. Efimov states evaded any experimental evidence for 35 years since their prediction until Kraemer
\textit{et al.} \cite{kraemer} reported on their discovery in an ultracold gas of Cesium atoms. 

In the literature it is claimed \cite{stecher,naturephysics,hammer} that few--body systems with $4 \leq N \leq 10$ identical particles 
show a universal behavior, namely the energies of the $n$-th level in three and $N$--body systems are universally related given that two--particle systems 
have a zero energy resonance and no negative energy bound states. Universal in this context means that such ratios are independent of the pair--interaction. 
The aim of the present letter is to show that such relation cannot hold exactly.  
We shall use the following operator notation \cite{jpa2}. $A\geq 0$ means that $\langle f| A|f\rangle \geq 0$ for all admissible $f$ and $A \ngeqq 0$
means that there exists $f_0$ such that $\langle f_0| A|f_0 \rangle < 0$.

Suppose that $N$ identical particles with mass $m$ in $\mathbb{R}^3$ interact through $V_0 (r_i - r_j)$, where $r_i$ for $i = 1, \ldots, N$ denotes particle's position vector and 
$V_0 (r) \leq 0$. We also require 
that $V_0 (r)$ is a bounded and finite range potential, that is there exists $r_V$ such that $V_0 (r) = 0$ for $|r| \geq r_V$ 
($V_0$ could be a spherical square well, for example). We choose the system of units, where $\hbar^2/m = 1$. The pair interaction is chosen in such a way that 
the following holds 
\begin{eqnarray}
-\Delta_r + V_0 (r) \geq 0 \label{31.1}\\
-\Delta_r + (1+\varepsilon)V_0 (r) \ngeqq 0 \quad \quad \textnormal{for all $\varepsilon >0$} . \label{31.2}
\end{eqnarray}
Eqs. (\ref{31.1})--(\ref{31.2}) mean that each pair of particles has a zero energy resonance (or is at critical coupling \cite{jpa2}). 
The ground state wave functions for $N=3$ and $N=4$ are determined through equations
\begin{eqnarray}
\left[ T^{(3)} + \sum_{1 \leq i < j \leq 3} V_0 (r_i - r_j) \right] \psi^{(3)}_0 = E^{(3)}_0 \psi^{(3)}_0   \label{31.3}\\
\left[ T^{(4)} + \sum_{1 \leq i < j \leq 4} V_0 (r_i - r_j) \right] \psi^{(4)}_0 = E^{(4)}_0 \psi^{(4)}_0 ,    \label{31.4}
\end{eqnarray}
where $T^{(3)}$,  $T^{(4)}$ are kinetic energy operators with the center of mass motion removed for $N = 3$ and $N = 4$ respectively. The wave functions 
$\psi^{(3)}_0$ and $\psi^{(4)}_0$ are normalized and symmetric with respect to the interchange of particles. 
By equation in \cite{stecher,naturephysics,hammer} 
\begin{equation}
 E^{(4)}_0 = C_u E^{(3)}_0 , \label{31.7}
\end{equation}
where the universal constant $C_u \simeq 4.6$ does not depend on the interaction potential $V_0$. Eq.~(\ref{31.7}) and the value of the universal constant were established 
through numerical calculations. Here we consider only ground states but the same argument can be applied to higher states if they are non--degenerate. 

Let us introduce the perturbation potentials and define $\chi_R , \eta_R : \mathbb{R}^3 \to \mathbb{R}$ so that $\chi_R (r) = 1$ for $|r| \in [0, R] $ and zero otherwise; 
$\eta_R (r) = 1$ for $|r| \in (R, 2R] $ and zero otherwise. So far we fix the parameter $R >0$ and define 
\begin{equation}
 V_\lambda (r) := V_0 (r) + \lambda \chi_R (r) - B(\lambda) \eta_R (r) . 
\end{equation}
The first coupling constant $\lambda \geq 0$ and the second coupling constant $B(\lambda)$ are determined through the equations 
\begin{eqnarray}
-\Delta_r + V_\lambda (r) \geq 0 \label{31.1a}\\
-\Delta_r + (1+\varepsilon)V_\lambda (r) \ngeqq 0 \quad \quad \textnormal{for all $\varepsilon >0$} . \label{31.2a}
\end{eqnarray}
That is, the perturbed potential retains the property that particle pairs have a zero energy resonance. It is easy to see that for each $\lambda \geq 0$ 
Eqs.~(\ref{31.1a})--(\ref{31.2a}) determine $B(\lambda) \geq 0$ uniquely. In Proposition~\ref{propo} we would show that in the vicinity of $\lambda =0$ the function $B(\lambda)$ 
has a well-defined first derivative $B'(\lambda):= dB /d\lambda$. For small $\lambda$ the perturbed versions of Eqs.~(\ref{31.3})--(\ref{31.4}) read 
\begin{eqnarray}
\left[T^{(3)} + \sum_{1 \leq i < j \leq 3} V_\lambda (r_i - r_j) \right] \psi^{(3)}_\lambda = E^{(3)}_\lambda \psi^{(3)}_\lambda   \label{31.5}\\
\left[T^{(4)} + \sum_{1 \leq i < j \leq 4} V_\lambda (r_i - r_j) \right] \psi^{(4)}_\lambda = E^{(4)}_\lambda \psi^{(4)}_\lambda .    \label{31.6}
\end{eqnarray}
Since the particle pairs interacting through $V_\lambda$ have a zero energy resonance Eq.~(\ref{31.7}) takes the form $ E^{(4)}_\lambda = C_u E^{(3)}_\lambda $. Differentiating 
both parts of this equation with respect to $\lambda$ and setting $\lambda = 0$ gives \cite{kato,reed,feynman}
\begin{eqnarray}
\fl  3C_u \langle \psi^{(3)}_0 | \chi_R (r_1 - r_2) | \psi^{(3)}_0 \rangle = 6 \langle \psi^{(4)}_0 | \chi_R (r_1 - r_2) |\psi^{(4)}_0 \rangle - 
6 B'(0) \langle \psi^{(4)}_0 | \eta_R (r_1 - r_2) |\psi^{(4)}_0 \rangle \nonumber\\
+ 3 C_u B'(0) \langle \psi^{(3)}_0 | \eta_R (r_1 - r_2) |\psi^{(3)}_0 \rangle . \label{31.11}
\end{eqnarray}
In (\ref{31.11}) we have used the permutation symmetry. Since $\psi^{(3)}_0$, $\psi^{(4)}_0$ are fixed square integrable functions for $R \to \infty$ we get 
\begin{eqnarray}
 \langle \psi^{(i)}_0 | \chi_R (r_1 - r_2) | \psi^{(i)}_0 \rangle = 1 + \hbox{o} (R) \quad \quad i = 3, 4 \label{1.11}\\
 \langle \psi^{(i)}_0 | \eta_R (r_1 - r_2) | \psi^{(i)}_0 \rangle = \hbox{o} (R) \quad \quad i = 3, 4. \label{1.12}
\end{eqnarray}
By Proposition~\ref{propo} from (\ref{31.11}), (\ref{1.11}), (\ref{1.12}) it follows that 
\begin{equation}
 3C_u = 6 + \hbox{o} (R) , 
\end{equation}
which means that $C_u = 2$ in clear contradiction with the numerical value $4.66$ in \cite{stecher} or $4.58$ in \cite{naturephysics}. 

We only need to prove the following technical result in the two--particle problem 
\begin{proposition}\label{propo}
$B(\lambda)$ is differentiable in the vicinity of $\lambda = 0$. Its derivative at zero  
$|B'(0)|$ remains bounded for $R \to \infty$. 
\end{proposition}
\begin{proof}
 We shall need the Birman-Schwinger (BS) principle in the form that is discussed in Appendix~A in \cite{gridnevvugalter}. 
Let us introduce the BS operator 
\begin{equation}
 D_0 :=  -\bigl(-\Delta_r\bigr)^{-1/2} V_0 \bigl(-\Delta_r\bigr)^{-1/2} . 
\end{equation}
By Eqs.~(\ref{31.1})--(\ref{31.2}) and the BS principle there exists 
$\varphi_0 \in L^2 (\mathbb{R}^3)$, $\|\varphi_0\| = 1$ such that the following equation holds 
\begin{eqnarray}
D_0 \varphi_0 = \varphi_0 \label{gov}\\
\sup \sigma(D_0) = \|D_0\|= 1, 
\end{eqnarray}
and $\sigma(A)$ always denotes the spectrum of the operator $A$. 
The operator $D_0$ is bounded and positivity improving, which means that for any $f \geq 0$ one also has $D_0 f \geq 0$ (for mathematical details see \cite{gridnevvugalter}). 
Thus by Theorem~XIII.43 in vol.~4 of \cite{reed} its eigenvalue 
equal to one is non--degenerate and $\varphi_0 >0$. Let us introduce the perturbed BS operator 
\begin{equation}
\fl D(\lambda, B) := D_0 -  \lambda \bigl(-\Delta_r\bigr)^{-1/2} \chi_R \bigl(-\Delta_r\bigr)^{-1/2} + B \bigl(-\Delta_r\bigr)^{-1/2} \eta_R \bigl(-\Delta_r\bigr)^{-1/2} , 
\end{equation}
which depends on the parameters $\lambda, B \in \mathbb{R}$. Let us define  
\begin{equation}
 \mu (\lambda , B) := \sup \sigma \bigl(D(\lambda, B)\bigr) , 
\end{equation}
where $\mu : \mathbb{R} \times \mathbb{R}  \to \mathbb{R} $ is jointly continuous in both arguments. Clearly, $\mu (0,0) = 1$ and by continuity 
for $\lambda, B$ around zero 
$\mu(\lambda, B)$ is the largest non-degenerate eigenvalue of $D(\lambda, B)$. By standard perturbation theory \cite{kato,reed} 
$\mu(\lambda, B)$ is analytic in $\lambda$ around 
$\lambda = 0$ and $B$ small and it is also analytic in $B$ around $B = 0$ and $\lambda$ small. 

We define the function $B(\lambda)$ through the equation 
\begin{equation}
  \mu (\lambda , B(\lambda)) = 1. \label{31.21a}
\end{equation}
In this case by the BS principle we guarantee that Eqs.~(\ref{31.1a})--(\ref{31.2a}) are satisfied. By the implicit function theorem (\ref{31.21a}) uniquely 
determines $B(\lambda)$ and $B'(\lambda)$ in the vicinity of $\lambda = 0$. The derivative can be explicitly calculated from (\ref{31.21a}) as follows 
\begin{equation}
 B'(0) = - \frac{\left(\frac{\ddp \mu }{\ddp \lambda}\right)_{\lambda, B = 0}}{\left(\frac{\ddp \mu }{\ddp B}\right)_{\lambda, B = 0}} = 
\frac{\langle \varphi_0 | \bigl(-\Delta_r\bigr)^{-1/2} \chi_R \bigl(-\Delta_r\bigr)^{-1/2} | \varphi_0 \rangle}{\langle \varphi_0 | \bigl(-\Delta_r\bigr)^{-1/2} \eta_R \bigl(-\Delta_r\bigr)^{-1/2} | \varphi_0 \rangle} \label{31.21}
\end{equation}
Let us define 
\begin{equation}
 \psi_0 (r) := \bigl(-\Delta_r\bigr)^{-1/2} \varphi_0 . 
\end{equation}
By standard results (see Sec.~XIII.11 in \cite{reed}) $\psi_0 (r)$ is a continuous function and due to positivity improving property of $\bigl(-\Delta_r\bigr)^{-1/2}$ we also have 
$\psi_0 (r) >0$. 
Below we shall prove that there exist constants $A_{1,2} >0$ such that the following bound holds 
\begin{equation}
 \frac{A_1}{(1+|r|)} \leq \psi_0 (r) \leq  \frac{A_2}{(1+|r|)} . \label{31.41}
\end{equation}
Let us estimate the nominator and the denominator in (\ref{31.21}) using (\ref{31.41}). For the nominator we get 
\begin{equation}
 \fl \langle \varphi_0 | \bigl(-\Delta_r\bigr)^{-1/2} \chi_R \bigl(-\Delta_r\bigr)^{-1/2} | \varphi_0 \rangle \leq 4\pi A_2^2 \int_0^R \frac{t^2 dt}{(1+t)^2} = 4\pi A_2^2 
\bigl[R- \ln R + \mathcal{O} (1) \bigr]. \label{31.42}
\end{equation}
And similarly 
\begin{equation}
 \fl \langle \varphi_0 | \bigl(-\Delta_r\bigr)^{-1/2} \eta_R \bigl(-\Delta_r\bigr)^{-1/2} | \varphi_0 \rangle \geq 4\pi A_1^2 \int_R^{2R} 
\frac{t^2 dt}{(1+t)^2} = 4\pi A_1^2 \bigl[R + \mathcal{O} (1) \bigr]. \label{31.43}
\end{equation}
Substitution of the estimates (\ref{31.42})--(\ref{31.43}) into (\ref{31.21}) finishes the proof. 
It remains to prove (\ref{31.41}). Due to continuity it suffices to prove that (\ref{31.41}) holds for $|r| \geq 2r_V$. By (\ref{gov}) 
\begin{equation}
 \psi_0 (r) = \frac 1{4\pi} \int \frac{V_0 (r') \psi_0 (r')}{|r-r'|} d^3r' , 
\end{equation}
which for $|r| \geq 2r_V$ gives 
\begin{equation}\label{1.1}
 \frac{C_0}{|r| - r_V} \leq \psi_0 (r) \leq \frac{C_0}{|r| + r_V} , 
\end{equation}
and where we defined $C_0 := (4\pi)^{-1}\int V_0 (r') \psi_0 (r') d^3 r' >0$. From (\ref{1.1}) it easily follows that one can always choose $A_{1,2}$ to make Eq.~(\ref{31.41}) hold. 
\end{proof}

In conclusion, let us remark that similar calculation shows that even in the case of the Efimov trimer 
the relation of the energies $E_n /E_{n+1}$ for a given $n$ is not constant for all pair-interactions; it holds universally only in the limit of large $n$! 
(In the case of large $n$ the argument presented in this paper breaks down at Eq.~(\ref{1.11}), because for large $n$ Efimov states are totally spreading in space, 
see Appendix in \cite{jpa2}). Such limit of large $n$ is unavailable for $N \geq 4$, see \cite{greenwood,gridnevvugalter}. 
Thus a relation (\ref{31.7}) can only hold approximately for a narrow class of pair--interactions. Hopefully, the result presented 
here would help to set such limits on Eq.~(\ref{31.7}) and analogous relations in the future.

\ack 

The author would like to thank Prof. Walter Greiner for the warm hospitality at FIAS.


\section*{References} 

\begin{thebibliography}{99}

\bibitem{vefimov} V. Efimov, Phys. Lett. \textbf{B 33}, 563 (1970); Sov. J. Nucl. Phys. \textbf{12}, 589 (1971)

\bibitem{sobolev} A. V. Sobolev, Commun. Math. Phys. \textbf{156}, 101 (1993)

\bibitem{braatenphysrep} E. Braaten, H.-W. Hammer, Phys. Rep. \textbf{428}, 259 (2006)

\bibitem{efimovprivate} V. Efimov, private communication (2012)

\bibitem{amadonoble} R. D. Amado and J. V. Noble, Phys. Rev. \textbf{D 5}, 1992 (1972); Phys. Lett. \textbf{35B}, 25 (1971)

\bibitem{yafaev} D. R. Yafaev, Math. USSR-Sb. \textbf{23}, 535 (1974)

\bibitem{ovchinnikov} Yu. N. Ovchinnikov, I. M. Sigal, Annals of Physics, \textbf{123}, 274 (1979)

\bibitem{fonseca} A. C. Fonseca, E. F. Redish, P. E. Shanley, Nucl. Phys. \textbf{A 320}, 273 (1978)

\bibitem{tamura} H. Tamura, J.~Funct.~Anal. \textbf{95} 433 (1991); Nagoya Math. J.  \textbf{130} 55 (1993)

\bibitem{wang}
X. P. Wang, J. Funct. Anal.~\textbf{209} (2004) pp. 137--161.

\bibitem{wangwrong}
D. K. Gridnev, J.~Funct.~Anal. \textbf{263}, 1485-1486 (2012)

\bibitem{kraemer}
T. Kraemer \textit{et al.}, Nature \textbf{440}, 315 (2006)

\bibitem{stecher}
J. von Stecher, J. Phys. B: At. Mol. Opt. Phys. \textbf{43}, 101002  (2010) 

\bibitem{naturephysics}
J. von Stecher, J. P. D`Incao and C. H. Greene, Nature Physics \textbf{5}, 417 (2009)

\bibitem{hammer}
H. W. Hammer and L. Platter, Eur. Phys. J. A \textbf{32}, 113 (2007)


\bibitem{jpa2}
D. K. Gridnev, J. Phys. A: Math. Theor. \textbf{45} 395302 (2012); arXiv:1112.0112v2

\bibitem{kato} T. Kato, {\em Perturbation Theory for Linear Operators}, Springer--Verlag/Berlin Heidelberg (1995) 

\bibitem{reed} M. Reed and B. Simon, {\em Methods of Modern
Mathematical Physics}, vol.~1 Academic Press/New York (1980), vo.~2 Academic Press/New York (1975) and vol.~4,
Academic Press/New York (1978).

\bibitem{feynman} R. P. Feynman, Phys. Rev. \textbf{56}, 340 (1939); H. Hellmann, \textit{Einführung in die Quantenchemie},  Leipzig: Franz Deuticke,  p. 285 (1937) 

\bibitem{gridnevvugalter} D. K. Gridnev, {\it Why there is no Efimov effect for four bosons and related results
on the finiteness of the discrete spectrum}, arXiv:1210.5147 [math-ph] 


\bibitem{greenwood} R. D. Amado, F. C. Greenwood, Phys. Rev. \textbf{D 7}, 2517 (1973).

\end{thebibliography}
\end{document}